\documentclass{amsart}
\usepackage{mathrsfs}
\usepackage{amssymb}
\usepackage{amsmath}
\usepackage{amsfonts}
\usepackage{enumerate}
\usepackage{pifont}
\usepackage{enumerate}
\usepackage[pdftex]{graphicx}
\usepackage{verbatim}
\usepackage{amsthm}
\usepackage[all]{xy}
\usepackage{xypic}

\numberwithin{equation}{section}

\newtheorem{definition}{Definition}[section]

\newtheorem{proposition}{Proposition}[section]
\newtheorem{remark}{Remark}[section]
\newtheorem{example}{Example}[section]

\newcommand{\8}{\infty}

\newcommand{\be}{\begin{eqnarray*}}
\newcommand{\ee}{\end{eqnarray*}}
\newcommand{\beq}{\begin{equation}}
\newcommand{\eeq}{\end{equation}}
\newcommand{\beqn}{\begin{equation*}}
\newcommand{\eeqn}{\end{equation*}}
\newcommand{\bs}{\begin{split}}
\newcommand{\es}{\end{split}}

\begin{document}

\title{Banach space formalism of quantum mechanics}

\author{Zeqian Chen}

\address{State Key Laboratory of Resonances and Atomic and Molecular Physics, Wuhan Institute of Physics and Mathematics, Innovation Academy for Precision Measurement Science and Technology, Chinese Academy of Sciences, 30 West District, Xiao-Hong-Shan, Wuhan 430071, China.}


\date{}
\maketitle
\markboth{Zeqian Chen}%
{Quantum theory}

\begin{abstract}
This paper presents a generalization of quantum mechanics from conventional Hilbert space formalism to Banach space one. We construct quantum theory starting with any complex Banach space beyond a complex Hilbert space, through using a basic fact that a complex Banach space always admits a semi-inner product. Precisely, in a complex Banach space $\mathbb{X}$ with a given semi-inner product, a pure state is defined by Lumer \cite{Lumer1961} to be a bounded linear functional on the space of bounded operators determined by a normalized element of $\mathbb{X}$ under the semi-inner product, and then the state space $\mathcal{S} (\mathbb{X})$ of the system is the weakly closed convex set spanned by all pure states. Based on Lumer's notion of the state, we associate a quantum system with a complex Banach space $\mathbb{X}$ equipped with a fixed semi-inner product, and then define a physical event at a quantum state $\omega \in \mathcal{S}(\mathbb{X})$ to be a projection $P$ (bounded operator such that $P^2 =P$) in $\mathbb{X}$ satisfying the positivity condition $0 \le \omega (P) \le 1,$ and a physical quantity at a quantum state $\omega$ to be a spectral operator of scalar type with real spectrum so that the associated spectral projections are all physical events at $\omega.$ The Born formula for measurement of a physical quantity is the natural pairing of operators with linear functionals satisfying the probability conservation law. A time evolution of the system is governed by a one-parameter group of invertible spectral operators determined by a scalar type operator with the real spectrum, which satisfies the Schr\"{o}dinger equation. Our formulation is just a generalization of the Dirac-von Neumann formalism of quantum mechanics to the Banach space setting. We include some examples for illustration.
\end{abstract}

\maketitle


\section{Introduction}\label{Intro}

Since the birth of quantum mechanics, various generalizations on it are developed (cf.\cite{Mac1963}). We refer to \cite{Nie2014, Nie2015, Sch2020, Sch2021} and references therein for some recent works. Those formulations essentially used the algebraic structure of state and observable, as similar to that in the conventional Hilbert space formalism (cf.\cite{Dirac1958}), and seem not to have a generalization to an arbitrary Banach space setting. Here, we shall use the analytic structure of state and observable (cf.\cite{vN1955}) to give a generalization of quantum mechanics from conventional Hilbert space formalism to Banach space one.

Based on the analytic structure of state and observable, the author \cite{Chen2021} presented a mathematical formalism of non-Hermitian quantum mechanics (cf.\cite{BBM1999, BBJ2002, Brody2014, Mosta2010}) by using spectral operators of scalar type with real spectrum in a complex Hilbert space. Since the theory of spectral operators has been studied in an arbitrary Banach space (cf.\cite{DS1971}), this stimulates us to generalize quantum theory to a Banach space setting. This can be done by using the notion of a state defined by Lumer \cite{Lumer1961} in a complex Banach space under a semi-inner product. Since a complex Banach space always admits a semi-inner product, our formulation of quantum mechanics is valid for any complex Banach space.

The paper is organized as follows. In Section \ref{Pre}, we include some definitions and notations on spectral operators and semi-inner products in a complex Banach space. In Section \ref{QevtQob}, we give the definitions of physical event and quantity and study their basic properties. In Section \ref{BSF}, we present a formulation of quantum mechanics in a Banach space setting, and give an example for illustration. Finally, we give a summary in Section \ref{Sum}.

\section{Preliminaries}\label{Pre}

We utilize the standard notions and notations from functional analysis (cf.\cite{DS1957, DS1963, RS1980I, Rudin1991}). We denote by $\mathbb{C}$ the complex field and the unit circle $\mathbb{S} = \{z \in \mathbb{C}: |z| =1 \}.$ Unless specified otherwise, we denote by $\mathbb{X}$ a complex Banach space. By an {\it operator} $T$ in $\mathbb{X}$ we shall mean a linear mapping whose domain $\mathcal{D} (T)$ is a (not necessarily closed) linear subset of $\mathbb{X}$ and whose range $\mathcal{R} (T)$ lies in $\mathbb{X}.$ If $\mathcal{D} (T)$ is dense in $\mathbb{X},$ then $T$ is called a densely defined operator, and if the set $\{(x, Tx): x \in \mathcal{D} (T) \}$ is closed in $\mathbb{X} \times \mathbb{X}$ equipped with the norm $\|(x,y)\| = \|x\| + \|y\|,$ then $T$ is called a closed operator. For a densely defined closed operator $T,$ we denote by $\sigma (T)$ its spectral set and by $\rho (T)$ its resolvent set. We always use $\mathcal{B} (\mathbb{X})$ to denote the space of all bounded operators in $\mathbb{X}.$

\subsection{Spectral operators}\label{Pre:SpecOper}

If an operator $P \in \mathcal{B} (\mathbb{X})$ satisfies $P^2 = P,$ it is called a {\it projection} as in \cite{DS1971,RS1980I,Rudin1991}. For a projection $P,$ its complementary operator $P^\bot = I -P$ is a projection as well. We denote by $\mathcal{P} (\mathbb{X})$ the set of all projections in $\mathbb{X}.$ For two commuting projections $P, Q \in \mathcal{P} (\mathbb{X}),$ the intersection $P\wedge Q$ is defined by
\be
P\wedge Q = P Q
\ee
with the range $P\wedge Q (\mathbb{X}) = P (\mathbb{X}) \cap Q (\mathbb{X}),$ and the union $P \vee Q$ by
\be
P \vee Q = P+Q - P Q,
\ee
with the range $P \vee Q (\mathbb{X}) = P (\mathbb{X}) + Q (\mathbb{X}) = \overline{\mathrm{span}} [P (\mathbb{X}) \cup Q (\mathbb{X})],$ the closed subspace of $\mathbb{X}$ spanned by the sets $P (\mathbb{X})$ and $Q (\mathbb{X}).$ The order $P \le Q$ between two commuting projections $P, Q \in \mathcal{P} (\mathbb{X})$ is defined to be $P (\mathbb{X}) \subset Q (\mathbb{X}).$ A Boolean algebra of projections in $\mathbb{X}$ is a subset of $\mathcal{P} (\mathbb{X})$ which is a Boolean algebra under operations $\wedge, \vee$ and $\le$ together with its zero and unit elements being the operators $0$ and $I$ in $\mathcal{B} (\mathbb{X})$ respectively.

\begin{definition}\label{df:SpecMeasure}{\rm (cf.\cite{Dunf1958, DS1971})}\;
Let $\Sigma$ be a $\sigma$-field of subsets of a non-empty set $\Omega.$ A spectral measure on $\Sigma$ is a map $\mathbf{E}$ from $\Sigma$ into a Boolean algebra of projections in $\mathbb{X}$ satisfying the following conditions:
\begin{enumerate}[$1)$]

\item $\mathbf{E} (\emptyset) =0$ and $\mathbf{E} (\Omega) =I.$

\item For any $A, B \in \Sigma,$ $\mathbf{E} (\Omega \setminus A) = \mathbf{E} (A)^\bot,$
\be
\mathbf{E} (A \cap B) = \mathbf{E} (A) \wedge \mathbf{E} (B), \quad \mathbf{E} (A \cup B) = \mathbf{E} (A) \vee \mathbf{E} (B).
\ee

\item $\mathbf{E} (A)$ is countably additive in $A$ in the strong operator topology, i.e., for every sequence $\{A_n\}$ of mutually disjoint sets in $\Sigma,$
\be
\mathbf{E} (\cup_n A_n) x = \sum_n \mathbf{E} (A_n)x
\ee
holds for any $x \in \mathbb{H},$ where the series of the right hand side converges in $\mathbb{X}$ in the norm topology.

\item $\mathbf{E}$ is regular, i.e., for any $x \in \mathbb{X}$ and $x^* \in \mathbb{X}^*,$ the complex-valued measure $A \mapsto x^* (\mathbf{E} (A)x)$ is regular on $\mathcal{B}_\mathbb{C}.$

\end{enumerate}
\end{definition}

Note that every spectral measure $\mathbf{E}$ on $\Sigma$ is bounded, i.e., $\sup_{A \in \Sigma} \| \mathbf{E} (A) \| < \8.$ In this case, the integral $\int_\Omega f(\omega) \mathbf{E} (d \omega)$ may be defined for every bounded $\Sigma$-measurable (complex-valued) function defined $\mathbf{E}$-almost everywhere on $\Omega.$ Recall that a function $f$ is defined $\mathbf{E}$-almost everywhere on $\Omega,$ if there exists $\Omega_0 \in \Sigma$ such that $\mathbf{E} (\Omega_0) = I$ and $f$ is well defined for every $\omega \in \Omega_0.$ It was shown (cf. \cite[X.1]{DS1963}) that this integral is a bounded mapping from the space $\mathcal{B}(\Omega, \Sigma)$ of bounded $\Sigma$-measurable functions in $\Omega$ with the norm $\|f\| = \sup_{\omega \in \Omega} |f(\omega)|$ into the Banach algebra $\mathcal{B} (\mathbb{X}),$ that is, for any $\alpha, \beta \in \mathbb{C}$ and for $f, g \in \mathcal{B}(\Omega, \Sigma),$
\be\begin{split}
\int_\Omega [\alpha f (\omega) + \beta g (\omega)] \mathbf{E} (d \omega) & = \alpha \int_\Omega f (\omega) \mathbf{E} (d \omega) + \beta \int_\Omega g (\omega) \mathbf{E} (d \omega),\\
\int_\Omega f (\omega) g (\omega) \mathbf{E} (d \omega) & = \int_\Omega f (\omega) \mathbf{E} (d \omega) \int_\Omega g (\omega) \mathbf{E} (d \omega),\\
\Big \| \int_\Omega f (\omega) \mathbf{E} (d \omega) \Big \| & \le C_\mathbf{E} \sup_{\omega \in \Omega} |f(\omega)|,
\end{split}\ee
where $C_\mathbf{E}$ is a positive constant depending only upon the spectral measure $\mathbf{E}.$

We shall focus on the spectral measures on the $\sigma$-field $\mathcal{B}_\mathbb{C}$ of Borel sets in the complex plane $\mathbb{C}.$
\begin{definition}\label{df:SpecOpScalar}{\rm (cf.\cite[Definition XVIII.2.12]{DS1971})}\;
A densely defined closed operator $T$ in $\mathbb{X}$ with the domain $\mathcal{D} (T)$ is of scalar type or called a scalar type operator, if there is a spectral measure $\mathbf{E}$ on $\mathcal{B}_\mathbb{C}$ such that
\be
\mathcal{D} (T) = \{ x \in \mathbb{X}: \lim_n T_n x\;\text{exists}\}
\ee
and
\be
T x = \lim_n  T_n x,\quad \forall x \in \mathcal{D} (T),
\ee
where
\be
T_n = \int_{\{z \in\mathbb{C}: |z| \le n\} } z \mathbf{E} (d z).
\ee
The spectral measure $\mathbf{E}$ is said to be the spectral resolution for $T.$
\end{definition}

\begin{remark}\label{rk:ScalarTypeOp}\rm
It is shown in \cite[Lemma XVIII.2.13]{DS1971} that a scalar type operator $T$ in the sense of Definition \ref{df:SpecOpScalar} is a spectral operator (cf.\cite[Definition XVIII.2.1]{DS1971}) and the spectral resolution of $T$ is unique. Thus, a scalar type operator is also called a {\it spectral operator of scalar type}. 
\end{remark}

\begin{definition}\label{df:FunctCalculusparaHop}{\rm (cf.\cite[Definition XVIII.2.10]{DS1971})}\;
Let $T$ be a scalar type operator of with the spectral resolution $\mathbf{E}$ on $\mathcal{B}_\mathbb{C}.$  For any $\mathcal{B}_\mathbb{C}$-measurable function $f,$ we define $f (T)$ by
\be
f(T) x = \lim_n T(f_n) x,\quad \forall x \in \mathcal{D} (f(T)),
\ee
where
\be\begin{split}
\mathcal{D} (f(T)) =& \{x \in \mathbb{H}: \lim_n T(f_n) x\;\text{exists}\},\\
T(f_n) = & \int_\mathbb{C} f_n (z) \mathbf{E} (d z),
\end{split}\ee
and
\be
f_n (z) = \left \{\begin{split} & f(z),\quad |f(z)| \le n,\\
& 0, \quad |f(z)| >0.
\end{split}\right.
\ee
\end{definition}

\begin{remark}\label{rk:FunctCalculusparaHop}\rm
It is shown in \cite[Theorem XVIII.2.17]{DS1971} that $f(T)$ in the above definition is a spectral operator of scalar type with the spectral resolution $\mathbf{E}_f (E) = \mathbf{E} (f^{-1} (E))$ for any $E \in \mathcal{B}_\mathbb{C}.$
\end{remark}

Thus, we have the well-defined functional calculus for scalar type operators, which plays a role in the dynamics of quantum mechanics in the Banach space setting. This is the reason why we use a scalar type operator as candidate for physical quantity in the Banach setting of quantum theory.

\subsection{Semi-inner product}\label{SIP}

Let $\mathbb{X}$ be a complex Banach space with the norm $\|\cdot\|,$ denoted by $(\mathbb{X}, \|\cdot\|).$ According to \cite{Lumer1961}, a semi-inner product on $(\mathbb{X}, \|\cdot\|)$ is a mapping $[\cdot,\cdot]: \mathbb{X} \times \mathbb{X} \mapsto \mathbb{C}$ satisfying the following properties:
\begin{enumerate}[$(1)$]

\item For any $x, y, z \in \mathbb{X}$ and for $a,b \in \mathbb{C},$
\be
[ax+by, z] = a [x,z]+ b [y,z].
\ee

\item $[x,x] =\|x \|^2$ for $x \in \mathbb{X}.$

\item For any $x,y \in \mathbb{X},$
\beq\label{eq:SIP-SchIeq}
|[x,y]|^2 \le [x,x] [y,y].
\eeq

\end{enumerate}

As shown in \cite{Lumer1961}, each complex Banach space has at least a semi-inner product. Indeed, for every $x \in \mathbb{X}$ there exists by the Hahn-Banach theorem at least one continuously linear functional $f_x \in \mathbb{X}^*$ with the norm $\|f_x\| = \|x\|$ such that $f_x (x) = \|x\|^2.$ For each $y \in \mathbb{X}$ we chose such a functional $f_y$ and define $[x,y] = f_y (x)$ for any $x \in \mathbb{X},$ and then get a semi-inner product $[\cdot,\cdot]$ on $(\mathbb{X}, \|\cdot\|).$ We remark that a Hilbert space $\mathbb{H}$ has a unique semi-inner product $[\cdot,\cdot]$ which is the Hilbert space inner product itself (cf.\cite[Theorem 3]{Lumer1961}).

\begin{example}\label{ex:2dimSIP}\rm
Consider a two-dimensional complex Banach space $(\mathbb{C}^2, \|\cdot\|_p)$ ($1 \le p < \8$), whose norm is defined by
\be
\|x\|_p = (|a|^p + |b|^p)^\frac{1}{p},\quad \forall x = a |0\rangle + b |1\rangle \in \mathbb{C}^2,
\ee
which is a Hilbert space only if $p=2,$ whereafter
\be
|0\rangle =  \left (
\begin{matrix}
1 \\
0
\end{matrix} \right ),
\quad
|1\rangle = \left (
\begin{matrix}
0 \\
1
\end{matrix} \right ).
\ee
We define $[\cdot, \cdot]_p: \mathbb{C}^2 \times \mathbb{C}^2 \mapsto \mathbb{C}$ by
\beq\label{eq:SIPqb1}
[x,y]_p = \|y\|^{2-p}_p (a \; \mathrm{sign} (\bar{c}) |c|^{p-1}+ b \; \mathrm{sign}(\bar{d}) |d|^{p-1}),
\eeq
for $x= a |0\rangle + b |1\rangle, y = c |0\rangle + d |1\rangle \in \mathbb{C}^2,$ where $\mathrm{sign}(u) = \frac{u}{|u|}$ if $u \not=0$ or $= 0$ if $u=0.$ Then $[\cdot, \cdot]_p$ is a semi-inner product on $(\mathbb{C}^2, \|\cdot\|_p).$
\end{example}

\section{Physical event and quantity}\label{QevtQob}

As noted above, every Banach space admits a semi-inner product. The notion of a state has been introduced by Lumer \cite{Lumer1961} in any complex Banach space under a semi-inner product.

\begin{definition}\label{df:State}{\rm (cf.\cite[Definition 8 and 9]{Lumer1961})}
Let $(\mathbb{X}, \|\cdot\|)$ be a complex Banach space with a fixed semi-inner product $[\cdot,\cdot].$ A bounded linear functional $\omega$ on $\mathcal{B}(\mathbb{X})$ is called a state if $\|\omega\|= \omega (I);$ if in addition $\|\omega\| =1,$ we shall call $\omega$ a normalized state. We shall denote by $\mathcal{S} (\mathbb{X})$ the set of all normalized states.

A state $\omega$ of the form $\omega (T) = [Tx, x]$ with a fixed $x \in \mathbb{X}$ is said to be a point state. Such a state is denoted by $\omega_x,$ sometimes written by $|x \rangle$ as well, and if $\|x\|=1$ we shall call $|x\rangle$ a pure state. We shall denote by $\mathcal{S}_0 (\mathbb{X})$ the set of all pure states.
\end{definition}

\begin{remark}\label{rk:state}\rm
As shown in \cite[Theorem 11]{Lumer1961}, $\mathcal{S} (\mathbb{X})$ is the weakly closed convex hull of $\mathcal{S}_0 (\mathbb{X}).$
\end{remark}

In what follows, based on Lumer's notion of the state, we introduce the notions of physical event and physical quantity on a complex Banach space equipped with a fixed semi-inner product.

\begin{definition}\label{df:Event}
Let $(\mathbb{X}, \|\cdot\|)$ be a complex Banach space with a fixed semi-inner product $[\cdot,\cdot].$ A projection $P$ on $(\mathbb{X}, \|\cdot\|)$ is called a physical event at a state $\omega \in \mathcal{S} (\mathbb{X}),$ if
\beq\label{eq:EventCond}
0 \le \omega (P) \le 1.
\eeq
\end{definition}

\begin{remark}\label{rk:event}\rm
In a complex Hilbert space $\mathbb{H}$ with the inner product $\langle \cdot, \cdot\rangle,$ an orthogonal projection is a physical event at any state. However, in a Banach space $\mathbb{X}$ with a semi-inner product, a projection $P$ being a physical event at a certain state $\omega$ may not be a physical event at another state $\omega'$ (see Example \ref{ex:PauliM} below), that is, a physical event is a state-dependent notion. This explains the reason why we define a physical event in a complex Banach space as above.
\end{remark}

\begin{definition}\label{df:HermitianOp}
Let $(\mathbb{X}, \|\cdot\|)$ be a complex Banach space with a fixed semi-inner product $[\cdot,\cdot].$ A scalar type operator $T$ with real spectrum on $(\mathbb{X}, \|\cdot\|)$ is called a physical quantity at a state $\omega,$ if for any Borel subset $A$ of $\mathbb{R},$ the associated spectral projection $\mathbf{E} (A)$ is a physical event at $\omega.$
\end{definition}

\begin{remark}\label{rk:HermitianOp}\rm
As similar to a physical event, a physical quantity is a state-dependent notion in the Banach space formalism of quantum mechanics.
\end{remark}

\begin{proposition}\label{prop:EventHermitian}\rm
Let $(\mathbb{X}, \|\cdot\|)$ be a complex Banach space with a fixed semi-inner product $[\cdot,\cdot].$
\begin{enumerate}[$1)$]

\item If a projection $P$ is a physical event at a state $\omega,$ then $P^\perp = 1-P$ is also a physical event at $\omega.$

\item If a scalar type operator $T$ with real spectrum has a normalized eigenvector $x$ with an isolated eigenvalue, then $T$ is a physical quantity at the quantum state $|x\rangle.$

\item If a scalar type operator $T$ with real spectrum is a physical quantity at a point state $\omega_x$ with $x \in \mathcal{D}(T),$ then the limit
\beq\label{eq:ExpHermitianOpPureState}
\omega_x (T) = \lim_n \omega_x (T_n)
\eeq
exists and is a real number, where $T_n = \int_{[-n, n]} \lambda \mathbf{E} (d \lambda)$ for $n \ge 1$ and $\mathbf{E}$ is the spectral measure associated with $T.$

\end{enumerate}
\end{proposition}

\begin{proof}
Statements $1)$ and $2)$ follow from Definitions \ref{df:Event} and \ref{df:HermitianOp}, while $3)$ from Definition \ref{df:SpecOpScalar}.
\end{proof}

\begin{remark}\label{prop:Basic}\rm
For a physical quantity $T$ at $\omega_x,$ if $x \in \mathcal{D}(T)$ and $\|x\|=1$ then $\omega_x (T)$ is the expectation of $T$ at the pure state $\omega_x$ or $|x\rangle,$ denoted by $\omega_x (T) = \langle x|T| x \rangle$ or $\langle T\rangle_x.$
\end{remark}

\begin{example}\label{ex:PauliM}\rm
As shown in Example \ref{ex:2dimSIP}, for $1 \le p < \8,$ the mapping $[\cdot, \cdot]_p: \mathbb{C}^2 \times \mathbb{C}^2 \mapsto \mathbb{C}$ defined by
\be\label{eq:SIPqb1}
[x,y]_p = \|y\|^{2-p}_p (a \; \mathrm{sign} (\bar{c}) |c|^{p-1}+ b \; \mathrm{sign}(\bar{d}) |d|^{p-1}),
\ee
for $x= a |0\rangle + b |1\rangle, y = c |0\rangle + d |1\rangle \in \mathbb{C}^2,$ is a semi-inner product on the two-dimensional complex Banach space $(\mathbb{C}^2, \|\cdot\|_p)$ with the norm $\|x\|_p = (|a|^p + |b|^p)^\frac{1}{p}.$

Consider the usual Pauli matrixes:
\beq\label{eq:PauliM}
\sigma_1  = \left (
\begin{matrix}
0 & 1 \\
1 & 0
\end{matrix} \right ),\;
\sigma_2  = \left (
\begin{matrix}
0 & - \mathrm{i} \\
\mathrm{i} & 0
\end{matrix} \right ),\;
\sigma_3  = \left (
\begin{matrix}
1 & 0 \\
0 & -1
\end{matrix} \right ).
\eeq
\begin{enumerate}[$1)$]

\item For $\sigma_1$ we have
\be
\sigma_1 = E^{\sigma_1}_x - E^{\sigma_1}_y,
\ee
where
\be\begin{split}
E^{\sigma_1}_x (z) =& \frac{1}{2} (u+v) (|0\rangle + |1\rangle),\\
E^{\sigma_1}_y (z) =&  \frac{1}{2} (u-v) (|0\rangle - |1\rangle),
\end{split}\ee
for any $z = u |0\rangle + v |1\rangle \in \mathbb{C}^2.$ Then
\be\begin{split}
[E^{\sigma_1}_x (z), z]_p =& \frac{1}{2} ( 1+ v \bar{u} |u|^{p-2} + u \bar{v} |v|^{p-2}),\\
[E^{\sigma_1}_y (z), z]_p =& \frac{1}{2} ( 1- v \bar{u} |u|^{p-2} - u \bar{v} |v|^{p-2}),
\end{split}\ee
for any unit vector $z = u |0\rangle + v |1\rangle \in \mathbb{C}^2.$ Hence, $\sigma_1$ is a physical quantity at a pure state $|z\rangle = u |0\rangle + v |1\rangle$ if and only if $u=0$ or $v=0$ or
\be
1 \pm \big ( \frac{v}{u} |u|^p + \frac{u}{v} |v|^p \big ) \ge 0,
\ee
where $|u|^p + |v|^p =1.$

\item For $\sigma_2$ we have
\be
\sigma_2 = E^{\sigma_2}_x - E^{\sigma_2}_y,
\ee
where
\be\begin{split}
E^{\sigma_2}_x (z) =& \frac{1}{2} (u- \mathrm{i} v) (|0\rangle + \mathrm{i}|1\rangle),\\
E^{\sigma_2}_y (z) =&  \frac{1}{2} (u+ \mathrm{i} v) (|0\rangle -\mathrm{i} |1\rangle),
\end{split}\ee
for any $z = u |0\rangle + v |1\rangle \in \mathbb{C}^2.$ Then
\be\begin{split}
[E^{\sigma_2}_x (z), z]_p =& \frac{1}{2} ( 1- \mathrm{i}v \bar{u} |u|^{p-2} + \mathrm{i}u \bar{v} |v|^{p-2} ),\\
[E^{\sigma_2}_y (z), z]_p =& \frac{1}{2} ( 1 + \mathrm{i}v \bar{u} |u|^{p-2} - \mathrm{i}u \bar{v} |v|^{p-2} ),
\end{split}\ee
for any unit vector $z = u |0\rangle + v |1\rangle \in \mathbb{C}^2.$ Hence, $\sigma_2$ is a physical quantity at a pure state $|z\rangle = u |0\rangle + v |1\rangle$ if and only if $u=0$ or $v=0$ or
\be
1 \pm \mathrm{i} \big ( \frac{v}{u} |u|^p - \frac{u}{v} |v|^p \big ) \ge 0,
\ee
where $|u|^p + |v|^p =1.$

\item For $\sigma_3$ we have
\be
\sigma_3 = E^{\sigma_3}_x - E^{\sigma_3}_y,
\ee
where $E^{\sigma_3}_x (z) = u |0\rangle , E^{\sigma_3}_y (z) = v |1\rangle$ for any $z = u |0\rangle + v |1\rangle \in \mathbb{C}^2.$ Then
\be
[E^{\sigma_3}_x(z), z]_p = \|z\|^{2-p}_p |u|^p \ge 0,\; [E^{\sigma_3}_y(z), z]_p = \|z\|^{2-p}_p |v|^p\ge 0.
\ee
Hence, $\sigma_3$ is a physical quantity at any pure state $|z\rangle.$
\end{enumerate}

In conclusion, we find that $\sigma_1$ is a physical quantity at some pure states but not so at some others, as well as $\sigma_2$ in the two-dimensional complex Banach space $(\mathbb{C}^2, \|\cdot\|_p)$ for $1 \le p\not= 2<\8;$ however, $\sigma_3$ is always a physical quantity at any state.
\end{example}

\section{Banach space formalism}\label{BSF}


Following the Dirac-von Neumann formalism of quantum mechanics \cite{Dirac1958, vN1955}, we present a Banach space formalism of quantum mechanics as follows.

\begin{definition}\label{df:MathFNHQM}
The Banach space formalism of quantum mechanics is one that a quantum system is associated with a complex Banach space $\mathbb{X}$ with a fixed semi-inner product $[\cdot,\cdot],$ satisfying the following four postulates:
\begin{enumerate}[$(P_1)$]

\item {\bf The state postulate}\; The system at any given time is described by a pure state $|x\rangle$ determined by a unit vector $x$ in $\mathbb{X}.$

\item {\bf The observable postulate}\; An observable for the system is a physical quantity $A$ at a certain pure state $|x\rangle.$

\item {\bf The measurement postulate}\; For a physical quantum $A$ at a pure state $|x\rangle,$ if $x \in \mathcal{D} (A)$ then the expectation of $A$ at $|x\rangle$ is given by
\beq\label{eq:BornRule}
\langle A \rangle_x = \omega_x (A),
\eeq
where $\omega_x (A)$ is defined in Proposition \ref{prop:EventHermitian}. In particular, if $A$ has a discrete spetrum $\{\lambda_n\}_{n \ge 1},$ whose eigenstates $\{e_n\}_{n \ge 1}$ is a unconditional basis in $\mathbb{X},$ then the expectation of $A$ at $|x\rangle$ with $x \in \mathcal{D} (A),$ where $x = \sum_{n \ge 1} e^*_n (x) e_n$ and $e^*_n$'s are the dual basis vectors of $e_n$'s such that $e^*_n(e_m) = \delta_{n m},$ is given by
\beq\label{eq:BornRuleV}
\langle A \rangle_x = \sum^\8_{n = 1} \lambda_n e^*_n (x) [e_n, x].
\eeq
In this case, $|x\rangle$ will be changed to the state $|e_n \rangle$ with probability
\beq\label{eq:TransProba}
p(x| e_n) = e^*_n (x) [e_n, x]
\eeq
for each $n \ge 1.$

\item {\bf The evolution postulate}\; Associated with a scalar type operator $H$ with real spectrum (the energy operator of the system), the time evolution operators $U(t) = e^{-\mathrm{i}t H}$ constitute an one-parameter group of invertible scalar type operators such that  $|x(t)\rangle = U(t) |x\rangle$ satisfies the Schr\"{o}dinger equation
\beq\label{eq:StateSchrEqu}
\mathrm{i} \frac{d | x (t)\rangle }{d t} = H |x(t)\rangle,
\eeq
with $|x(0)\rangle =  |x\rangle.$


\end{enumerate}
\end{definition}

\begin{remark}\label{rk:NHQM}\rm
\begin{enumerate}[$1)$]

\item In case $\mathbb{X}$ is a Hilbert space, these postulates coincide with the Dirac-von Neumann formalism of quantum mechanics, since the semi-inner product is the Hilbert space inner product.

\item We do not include the composite-systems postulate in the formalism. This needs Banach space tensor product theory (cf. \cite{DU1977}), and will discuss it elsewhere.

\item Since the spectral measure $\mathbf{E}$ of a scalar type operator with real spectrum is countably additive and $\mathbf{E} (\mathbb{R}) = I,$ the measurement postulate $(P_3)$ satisfies the probability conservation law. In particular, by \eqref{eq:TransProba} one has
\be
\sum_n p(x| e_n) = \sum_n e^*_n (x) [e_n, x] = [x,x] =1,
\ee
this concludes the probability conservation.

\item As in the case of non-Hermitian quantum mechanics (cf.\cite{Brody2014}), the evolution of state in \eqref{eq:StateSchrEqu} does not in general preserve the Banach space norm of state vectors.

\end{enumerate}

\end{remark}

The following example shows Banach space formalism of quantum mechanics for a qubit.

\begin{example}\rm
Consider a quantum system associated with the two-dimensional complex Banach space $(\mathbb{C}^2, \|\cdot\|_p)$ ($1 \le p < \8$). We shall construct quantum theory in $(\mathbb{C}^2, \|\cdot\|_p).$ As in Example \ref{ex:2dimSIP}, recall that $[\cdot, \cdot]_p: \mathbb{C}^2 \times \mathbb{C}^2 \mapsto \mathbb{C}$ defined by
\be\label{eq:SIPqb1}
[x,y]_p = \|y\|^{2-p}_p (a \; \mathrm{sign} (\bar{c}) |c|^{p-1}+ b \; \mathrm{sign}(\bar{d}) |d|^{p-1}),
\ee
for $x= a |0\rangle + b |1\rangle, y = c |0\rangle + d |1\rangle \in \mathbb{C}^2,$ is a semi-inner product on $(\mathbb{C}^2, \|\cdot\|_p).$

For two unit vectors $x= a |0\rangle + b |1\rangle, y = c |0\rangle + d |1\rangle \in \mathbb{C}^2$ such that $a d - c b \not=0,$ we can define two projections $E_x$ and $E_y$ by
\be\begin{split}
E_x (z) = & \frac{1}{a d - c b}( d u - c v) x,\\
E_y (z) = & \frac{1}{a d - c b}(- b u + a v) y,
\end{split}\ee
for all $z = u |0\rangle + v |1\rangle \in \mathbb{C}^2.$ Then,
\be
E_x + E_y =I.
\ee
Therefore, a scalar type operator $H$  of the form
\be
H = \lambda_1 E_x + \lambda_2 E_y
\ee
with $\lambda_1, \lambda_2 \in \mathbb{R},$ is a physical quantity at a pure state $|z\rangle$ if and only if
\be
[E_x(z), z]_p \ge 0, [E_y(z), z]_p \ge 0.
\ee
(Note that $H$ is always a physical quantity at both $|x\rangle$ and $|y\rangle.$) In this case, the expectation of $H$ at $|z\rangle$ is given by
\be
\langle H\rangle_z = \lambda_1 [E_x z,z]_p +  \lambda_2 [E_y z,z]_p,
\ee
where $[E_x z,z]_p$ (resp. $[E_y z,z]_p$) is the probability obtaining $\lambda_1$ (resp. $\lambda_2$) under measurement at $|z\rangle.$

For an energy operator $H$ of the form $H = \lambda_1 E_x + \lambda_2 E_y$ with real spectrum, the time evolution $U(t) = e^{- \mathrm{i} t H} = e^{-\mathrm{i}t \lambda_1} E_x + e^{-\mathrm{i} t\lambda_2} E_y$ satisfies the Schr\"{o}dinger equation
\be
\mathrm{i} \frac{d U(t)}{d t} = H U(t)
\ee
with $U(0) =I,$ such that the state evolution $|x(t)\rangle = U(t) |x_0\rangle$ satisfies the usual Schr\"{o}dinger equation
\be
\mathrm{i} \frac{d |x(t)\rangle}{d t} = H |x(t)\rangle
\ee
with $|x(0)\rangle = |x_0\rangle \in \mathbb{C}^2 \setminus \{0\}.$

In conclusion, we get a quantum theory on $(\mathbb{C}^2, \|\cdot\|_p)$ with the semi-inner product $[\cdot,\cdot]_p.$
\end{example}

\section{Summary}\label{Sum}

Based on Lumer's notion of a state \cite{Lumer1961}, we introduce the notions of physical event and quantity for a quantum system associated with a complex Banach space equipped with a semi-inner product. Using these notions, we present a Banach space formulation of quantum mechanics, including the four postulates: the state postulate, the observable postulate, the measurement postulate, and the evolution postulate. These postulates generalize the conventional Hilbert space formalism of quantum mechanics to Banach space one. We include some examples to illustrate the physical meaning of those notions in the Banach space setting. It may not be unreasonable to hope that this formulation could have heuristic value to quantum physics in practice.

\


\bibliography{apssamp}

\end{document}